\documentclass{article}
\usepackage[margin=1in]{geometry}

\usepackage[utf8]{inputenc}
\usepackage{subfigure}

\usepackage[round]{natbib}

\author{Yury Kartynnik
  \and Andrew Ryzhikov}
\title{On Minimum Maximal Distance-$k$ Matchings}

\usepackage{amsmath}
\usepackage{amsthm}
\usepackage{amssymb}

\newtheorem{theorem}{Theorem}
\newtheorem{lemma}{Lemma}

\newtheorem{corollary}{Corollary}
\newtheorem{claim}{Claim}
\newtheorem{propos}{Proposition}

\usepackage{tikz}
\usetikzlibrary{fit,graphs,graphs.standard,quotes,shapes}

\begin{document}
\maketitle
\begin{abstract}
    We study the computational complexity of several problems connected with
    finding a maximal distance-$k$ matching of minimum cardinality or minimum
    weight in a given graph. We introduce the class of $k$-equimatchable graphs
    which is an edge analogue of $k$-equipackable graphs. We prove that the
    recognition of $k$-equimatchable graphs is co-NP-complete for any fixed \begin{math}k
    \ge 2\end{math}. We provide a simple characterization for the class of strongly
    chordal graphs with equal $k$-packing and $k$-domination numbers. We also
    prove that for any fixed integer \begin{math}\ell \ge 1\end{math} the problem of finding a
    minimum weight maximal distance-$2\ell$ matching and the problem of finding
    a minimum weight \begin{math}(2 \ell - 1)\end{math}-independent dominating set are NP-hard to
    approximate in chordal graphs within a factor of \begin{math}c \ln |V(G)|\end{math}, where $c$
    is a fixed constant (thereby improving the NP-hardness result of Chang for
    the independent domination case). Finally, we show the NP-hardness of the
    minimum maximal induced matching and independent dominating set problems in
    large-girth planar graphs.

    \textit{Note:} This version (as compared to the journal submission) contains corrections to Section 4.
\end{abstract}

\section{Introduction}

In this paper we consider finite undirected graphs without loops and multiple edges. Let $G$ be a graph with the vertex set $V(G)$ and the edge set $E(G)$. By $L(G)$ we denote the line graph of $G$, and by $G^k$ we denote the $k$-th power of $G$.
The distance between a vertex \begin{math}v \in V(G)\end{math} and an edge \begin{math}e \in E(G)\end{math} in a graph $G$ is the length of a shortest path between $v$ and any end vertex of $e$, i.e. \begin{math}\mathrm{dist}_G(v, e) = \min \{\mathrm{dist}_G(v, u), \mathrm{dist}_G(v, w)\}\end{math} for \begin{math}e = uw\end{math}. The distance between two edges \begin{math}e_1, e_2 \in E(G)\end{math} is the minimum of the distances between $e_1$ and any end vertex of $e_2$, i.e. \begin{math}\mathrm{dist}_G(e_1, e_2) = \min_{u \in e_1, v \in e_2} \{\mathrm{dist}_G(u, v)\}\end{math}. For other common definitions, see \cite{Bondy}.

A set $P$ of vertices in a graph is called a {\em $k$-packing} [\cite{MeirMoon}] if the distance between any two distinct vertices in this set is larger than $k$. A set $D$ of vertices in a graph $G$ is called a {\em $k$-dominating set} if for any vertex $v$ in $V(G)$ there is a vertex in $D$ at a distance no more than $k$ from $v$. A set $I$ of vertices in a graph is called a {\em $k$-independent dominating set} if it is both a $k$-packing and a $k$-dominating set. A graph is called {\em $k$-equipackable} if all its maximal $k$-packings are of the same size, i.e. \begin{math}i_k(G) = \alpha_k(G)\end{math}.

A set $M$ of edges in a graph $G$ is a {\em distance-$k$ matching} if all pairwise distances between distinct edges in $M$ are at least $k$. This notion was introduced by \cite{StockmeyerVazirani} under the name of $k$-separated matchings. A distance-$1$ matching is usually called simply a {\em matching}, and a distance-$2$ matching is also called an {\em induced matching}. We call a distance-$k$ matching {\em maximal} if it is not contained in any other distance-$k$ matching, and {\em maximum} if it has maximum cardinality. A graph is called {\em $k$-equimatchable} if all its maximal distance-$k$ matchings are of the same size. Obviously, a graph is $k$-eqiumatchable if and only if its line graph is $k$-equipackable.

Since a distance-$k$ matching in a graph $G$ corresponds to a $k$-packing in its line graph $L(G)$, and therefore to an independent set in $L(G)^k$, problems considering distance-$k$ matchings are special cases of problems considering $k$-packings and, in turn, independent sets. Several non-approximability results for maximum $k$-packings in bipartite and chordal graphs are presented in \cite{DistInd2014}. In \cite{DistIndDom} the problem of finding a $k$-independent dominating set of minimum size is proved to be NP-hard for bipartite graphs, and in \cite{McRaeHedetniemi} it is proved that there is no polynomial time $c$-approximation algorithm for this problem for any constant $c$ unless \begin{math}\mathrm{P} = \mathrm{NP}\end{math}.
\cite{Chang} has shown that the weighted distance-$1$ version of the problem is NP-hard in the class of chordal graphs, while the unweighted case of this problem is known to be solvable in polynomial time [\cite{Far82}]. For a survey on the complexity of the independent dominating set problem in various graph classes, see \cite{OGdW09}.

The \textsc{Minimum Maximal Distance-$k$ Matching} problem is to find a maximal distance-$k$ matching of the minimum size in a given graph $G$. \cite{MMM-Orlovich} mention some prior results on the complexity aspects of \textsc{Minimum Maximal Induced Matching} problem (which is an instance of \textsc{Minimum Maximal Distance-$k$ Matching} for the case of $k=2$) in various graph classes and prove that this problem is NP-hard to approximate in bipartite graphs within a factor of $n^{1-\varepsilon}$ for any constant \begin{math}\varepsilon > 0\end{math}, where $n$ is the number of vertices of the input graph.

Dualities between different independence-related parameters are also widely studied in the literature. The most widely studied class is the class of $1$-equipackable graphs, also known as well-covered graphs [\cite{PlummerCoverings}]. The recognition of such graphs is a co-NP-complete problem [\cite{ChvatalNote}], even in the class of \begin{math}K_{1, 4}\end{math}-free graphs [\cite{CaroGreedy}]. The recognition of $k$-equipackable graphs is also co-NP-complete for arbitrary fixed $k$ [\cite{IndDomSeq}].

The edge analogue of $k$-equipackable graphs is the class of $k$-equimatchable graphs. For the case of \begin{math}k = 1\end{math}, this class was introduced by \cite{equimatch} along with an implicit characterization which leads to a polynomial time recognition algorithm. See \cite{Plummer} for a survey of further related results. In \cite{Orlovich} it is proved that the recognition of $2$-equimatchable graphs is a co-NP-complete problem.

The maximum size of a $k$-packing in a graph $G$ is called the {\em $k$-packing number} of $G$ and is denoted by $\rho_k(G)$. Clearly, \begin{math}\rho_k(G) = \rho_1(G^k)\end{math}, where \begin{math}\rho_1(H) = \alpha(H)\end{math} is the classical independence number of the graph $H$.
The minimum size of a $k$-dominating set in a graph $G$ is called the {\em $k$-domination number} of $G$ and is denoted by $\gamma_k(G)$.
The minimum size of a $k$-independent dominating set in a graph $G$ is called the {\em $k$-independent domination number} of $G$ and is denoted by $i_k(G)$. It is easy to see that for every graph $G$, the inequalities \begin{math}\gamma_k(G) \le i_k(G) \le \rho_k(G)\end{math} hold (with \begin{math}\gamma_k(G) = \gamma_1(G^k)\end{math} and \begin{math}i_k(G) = i_1(G^k)\end{math} as well).

Dualities between independence and domination parameters comprise another important field of study. In \cite{MeirMoon} it is proved that the $k$-packing and $2k$-domination numbers equal for every tree. In \cite{ToppVolkman} the same equality is proved for block graphs, and in \cite{sun-free} this result is extended to the class of strongly chordal (also known as sun-free chordal [\cite{strongly-chordal}]) graphs, which includes powers of interval graphs and powers of block graphs. The trees with equal $k$-packing and $k$-domination numbers are characterized by \cite{ToppVolkman}. The graphs with equal $k$-packing and $2k$-packing numbers are characterized by \cite{DistPack}.

The paper is organized as follows. In Section \ref{k-equim} we show that the recognition of $k$-equimatchable graphs is co-NP-complete, generalizing the result of \cite{Orlovich} from $k=2$ to the case of an arbitrary \begin{math}k \ge 2\end{math}. It also follows from the proof that the problem of finding a minimum maximal distance-$k$ matching in a graph is NP-hard for every fixed \begin{math}k \ge 2\end{math}. In Section \ref{k-equip} we show that the class of strongly chordal graphs with equal maximum $k$-packing and $k$-dominating set sizes has a simple polynomial time characterization and show the connection of this result to the problem of characterizing $k$-equipackable graphs. In Section \ref{non-approx} we show that the problem of finding a minimum weight distance-$k$ matching is NP-hard to approximate in the class of chordal graphs within a factor of \begin{math}c \ln n\end{math}, where $n$ is the number of the input graph vertices, for some fixed constant $c$ and every fixed even $k$. We also establish a similar result for the problem of finding a minimum weight $k$-independent dominating set for odd $k$, thus strengthening the result of \cite{Chang} and of \cite{McRaeHedetniemi}. In Section \ref{sec-restr} we show that the problem of finding a minimum maximal distance-$2$ matching is NP-hard for planar graphs $G$ of girth lower bounded by \begin{math}O(|V(G)|^{\frac13 - \varepsilon})\end{math}, and the problem of finding a minimum independent dominating set is NP-hard for planar graphs $G$ of degree at most $5$ and girth lower bounded by \begin{math}O(|V(G)|^{1 - \varepsilon})\end{math}. We conclude with some open questions in Section \ref{sec-concl}.

\section{Recognition of $k$-equimatchable graphs}\label{k-equim}

In this section we prove that the recognition of $k$-equimatchable graphs is a co-NP-complete problem for any fixed \begin{math}k \ge 2\end{math} by constructing a polynomial time reduction from an NP-hard \textsc{$3$SAT} problem [\cite{Garey}]. We generalize the reduction proposed by \cite{Orlovich} for $2$-equimatchable graphs, which is in turn a generalization of the technique proposed by \cite{ChvatalNote}.

Let $X$ be a subset of vertices of $G$. By $G[X]$ we denote the subgraph of $G$ induced by $X$. By $E^t_G(X)$ we denote the set of all edges in $G$ at a distance less than $t$ from some vertex in $X$. For example, $E^1_G(\{v\})$ is the set of all the edges incident to the vertex $v$.

Consider the following polynomial time algorithm \textsc{SAT-to-Graph$_k$} that constructs a graph $G_{k}$ given an integer number \begin{math}k \ge 2\end{math} and a multiset of clauses \begin{math}C = \{c_1, c_2, \ldots, c_m\}\end{math}, \begin{math}m > 1\end{math}, over a set of boolean variables \begin{math}X = \{x_1, x_2, \ldots, x_n\}\end{math}, \begin{math}n > 1\end{math}. We shall assume that no clause contains a variable and its negation at the same time.

For every clause \begin{math}c_i \in C\end{math}, add a vertex $c'_i$ to $V(G_k)$. Let these vertices induce a clique $C'$ in $G_k$ by adding the corresponding edges.
Proceed depending on whether $k$ is even or odd (see Figure \ref{kr.fig.sat-to-graph} for the illustration of the cases \begin{math}k = 4\end{math} and \begin{math}k = 5\end{math}).

\textbf{Case 1:} \begin{math}k = 2\ell\end{math} for some integer \begin{math}\ell \ge 1\end{math}. For every variable $x_i$, add a set of new vertices \begin{math}X_i \cup \overline{X}_i\end{math} to $V(G_k)$, where \begin{displaymath}X_i = \{x_j^i: 1 \le j \le \ell + 1\},\end{displaymath} \begin{displaymath}\overline{X}_i = \{\overline{x}_j^i: 1 \le j \le \ell + 1\}.\end{displaymath}
Make \begin{math}(x_1^i, x_2^i, \ldots, x_{\ell + 1}^i, \overline{x}_{\ell + 1}^i, \overline{x}_\ell^i, \ldots, \overline{x}_1^i)\end{math} induce a simple path for every $i$ by adding the corresponding edges.

\textbf{Case 2:} \begin{math}k = 2\ell + 1\end{math} for some integer \begin{math}\ell \ge 1\end{math}. For every variable $x_i$, add a set of new vertices \begin{math}X_i \cup \overline{X}_i \cup \{a^i, b^i\}\end{math} to $V(G_k)$, where \begin{displaymath}X_i = \{x_1^i\} \cup \{x_j^i, x_j'^i: 2 \le j \le \ell + 1\},\end{displaymath} \begin{displaymath}\overline{X}_i = \{\overline{x}_1^i\} \cup \{\overline{x}_j^i, \overline{x}_j'^i: 2 \le j \le \ell + 1\}.\end{displaymath}
Make \begin{math}(x_1^i, x_2^i, \ldots, x_{\ell+1}^i, a^i, \overline{x}_{\ell+1}^i, \overline{x}_\ell^i, \ldots, \overline{x}_1^i)\end{math} induce a simple path for every $i$ by adding the corresponding edges. Add the edges $a^ib^i$, \begin{math}x_j^i x_j'^i\end{math}, \begin{math}\overline{x}_j^i \overline{x}_j'^i\end{math} to $E(G_k)$ for \begin{math}2 \le j \le \ell + 1\end{math} and every $i$.

Besides that, in both cases, for every occurrence of the literal $x_i$ (respectively, $\overline{x}_i$) in the clause $c_j$, add a simple \begin{math}(x_1^i, c'_j)\end{math}-path (respectively, \begin{math}(\overline{x}_1^i, c'_j)\end{math}-path) of length $\ell$, introducing new intermediate vertices when necessary.

\noindent\begin{figure}[!ht]
\centering
\usetikzlibrary{fit,graphs}
\def\nx{\overline{x}}
\begin{tikzpicture}[
    vertex/.style={circle, draw=black, fill, minimum width=1.5mm, inner sep=0pt, node distance=0pt},
    every label/.style={inner sep=2pt},
    baseline=(current bounding box.north),
    scale=1.2
]
\graph[nodes=vertex, empty nodes, no placement] {
    {
        [y=1.5]
        c1[x=-2,label={[name=c1label]left:$c'_1$}];
        cj[x=0,label=right:$c'_j$];
        cm[x=2,label={[name=cmlabel]right:$c'_m$}];
    };
    {
        [x=-0.7]
        xi1[y=0,label={left:$x^i_1$}] --
        xi2[y=-0.5,label={left:$x^i_2$}] --
        xi3[y=-1,label={left:$x^i_3$}];
    };
    {
        [x=0.7]
        nxi1[y=0,label={right:$\nx^i_1$}] --
        nxi2[y=-0.5,label={right:$\nx^i_2$}] --
        nxi3[y=-1,label={right:$\nx^i_3$}];
    };
    xi3 -- nxi3;
    xi1 -- c1;
    xi1 -- cj;
    nxi1 -- cm;
};

\draw (xi1) -- node[vertex,midway] {} (c1);
\draw (xi1) -- node[vertex,midway] {} (cj);
\draw (nxi1) -- node[vertex,midway] {} (cm);

\node[rectangle,dotted,gray,draw,fit=(c1)(c1label)(cm)(cmlabel),rounded corners=2mm,inner sep=2pt] {};
\node at (-1,1.5) {$\ldots$};
\node at (1,1.5) {$\ldots$};

\end{tikzpicture}
\hskip 3em
\usetikzlibrary{fit,graphs}
\def\nx{\overline{x}}
\begin{tikzpicture}[
    vertex/.style={circle, draw=black, fill, minimum width=1.5mm, inner sep=0pt, node distance=0pt},
    every label/.style={inner sep=2pt},
    baseline=(current bounding box.north),
    scale=1.2
]
\graph[nodes=vertex, empty nodes, no placement] {
    {
        [y=1.5]
        c1[x=-2,label={[name=c1label]left:$c'_1$}];
        cj[x=0,label=right:$c'_j$];
        cm[x=2,label={[name=cmlabel]right:$c'_m$}];
    };
    {
        [x=-0.7]
        xi1[y=0,label={right:$x^i_1$}] --
        xi2[y=-0.5,label={right:$x^i_2$}] --
        xi3[y=-1,label={below:$x^i_3$}];
    };
    {
        [x=-1.2]
        xi2 -- xpi2[y=-0.5,label={left:$x^{\prime i}_2$}];
        xi3 -- xpi3[y=-1,label={left:$x^{\prime i}_3$}];
    };
    {
        [x=0.7]
        nxi1[y=0,label={left:$\nx^i_1$}] --
        nxi2[y=-0.5,label={left:$\nx^i_2$}] --
        nxi3[y=-1,label={below:$\nx^i_3$}];
    };
    {
        [x=1.2]
        nxi2 -- nxpi2[y=-0.5,label={right:$\nx^{\prime i}_2$}];
        nxi3 -- nxpi3[y=-1,label={right:$\nx^{\prime i}_3$}];
    };
    xi3 -- ai[x=0,y=-1,label={above:$a^i$}] -- nxi3;
    ai -- bi[x=0,y=-1.5,label={below:$b^i$}];
    xi1 -- c1;
    xi1 -- cj;
    nxi1 -- cm;
};

\draw (xi1) -- node[vertex,midway] {} (c1);
\draw (xi1) -- node[vertex,midway] {} (cj);
\draw (nxi1) -- node[vertex,midway] {} (cm);

\node[rectangle,dotted,gray,draw,fit=(c1)(c1label)(cm)(cmlabel),rounded corners=2mm,inner sep=2pt] {};
\node at (-1,1.5) {$\ldots$};
\node at (1,1.5) {$\ldots$};

\end{tikzpicture}
    \caption{Gadgets produced by \textsc{SAT-to-Graph$_k$}: $k = 4$ (left) and $k = 5$,  assuming $x_1 \in c_1, x_1 \in c_j, \overline{x}_1 \in c_m$. Dashed lines enclose cliques.\label{kr.fig.sat-to-graph}}
\end{figure}
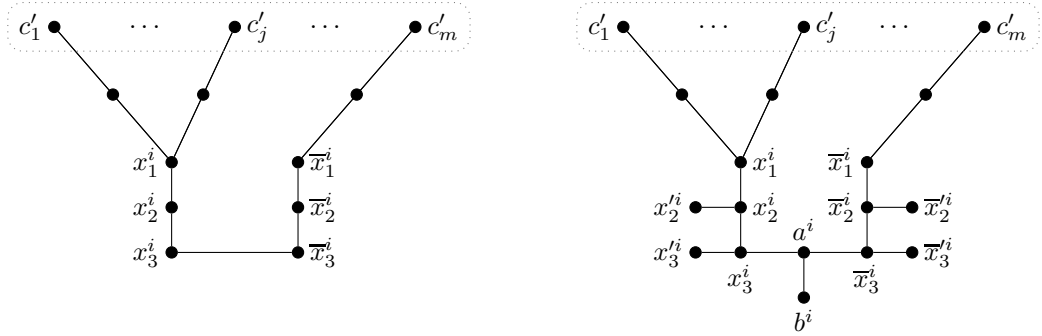

It is easy to see that the algorithm \textsc{SAT-to-Graph$_k$} runs in polynomial time. Below we investigate the properties of the graph $G_k$ constructed by this algorithm.
For an even $k$, denote by $V_i$ the union of the sets $X_i$, $\overline{X}_i$.
For an odd $k$, denote by $V_i$ the union of the sets $X_i$, $\overline{X}_i$, \begin{math}\{a^i, b^i\}\end{math}.

\begin{claim}\label{two-sizes} For every integer number \begin{math}k \ge 2\end{math}, the cardinality of any maximal distance-$k$ matching in $G_k$ is equal to $n$ or \begin{math}n + 1\end{math}, where \begin{math}n = |X|\end{math} is the number of variables in $X$. Moreover, if \begin{math}|M| = n\end{math} for some maximal distance-$k$ matching $M$ in $G_k$, then each edge of $M$ belongs to some $E(G_k[V_i])$.
\end{claim}

\begin{proof}
Note that the edge sets $E(G_k[V_i])$, \begin{math}i = \overline{1,n}\end{math}, and $E^l_{G_k}(C')$ partition $E(G)$. The edge set $E^\ell_{G_k}(C')$ induces a graph with the diameter \begin{math}2\ell + 1\end{math}, \begin{math}\ell = \lfloor k/2 \rfloor\end{math}, and thus this set cannot have more than one edge from the set $M$. For each \begin{math}1 \le i \le n\end{math}, the diameter of the graph induced by $V_i$ is equal to \begin{math}k + 1\end{math}, so this subgraph cannot have more than one edge from $M$. This covers all the edges of $G_k$, hence \begin{math}|M| \le n + 1\end{math}.

On the other hand, if for some $i$ the edge set $E^\ell_{G_k}(V_i)$ does not contain any edges of $M$, we can add to $M$ the edge $a^ib^i$ for an odd $k$ or the edge \begin{math}x_{\ell+1}^i \overline{x}_{\ell+1}^i\end{math} for an even $k$, and $M$ will still be a distance-$k$ matching. Hence, \begin{math}|M| \ge n\end{math}.

Next, let $M$ be a maximal distance-$k$ matching in $G_k$ such that \begin{math}|M| = n\end{math}. We shall prove that each edge of $M$ belongs to some of the graphs $G_k[V_i]$, $i=\overline{1,n}$. Assume the contrary: there is an edge $e$ that does not belong to $G_k[V_i]$ for any $i$. This means that \begin{math}e \in E^{\ell}_{G_k}(C')\end{math}. As \begin{math}|M| = n\end{math}, there is such an $i$ that $G_k[V_i]$ does not contain any edges of $M$. Recall that $E^{\ell}_{G_k}(V_i)$ should still contain an edge $f$ of $M$. This edge is not in $E(G_k[V_i])$, so \begin{math}f \in E^{\ell}_{G_k}(C')\end{math}. As $E^{\ell}_{G_k}(C')$ cannot contain more than one edge of $M$, we have \begin{math}f = e\end{math}.

Assume without loss of generality that $x^i_1$ is closer to $e$ than $\overline{x}^i_1$ and denote the distance from $e$ to $x_1^i$ as $d$. If $k$ is even, let \begin{math}e' = \overline{x}^i_{d+1} \overline{x}^i_{d+2}\end{math}, else let \begin{math}e' = \overline{x}^i_{d+2} \overline{x}_{d+2}'^i\end{math}. Recall that no clause contains both a variable and its negation, so \begin{math}M \cup \{e'\}\end{math} is a distance-$k$ matching. This contradicts the maximality of $M$.
\end{proof}

\begin{claim}\label{sizenp1} There is a distance-$k$ matching of size \begin{math}n + 1\end{math} in $G_k$ for any \begin{math}k \ge 2\end{math}.
\end{claim}

\begin{proof}
For an even $k$, the edge set  \begin{math}\{x_{\ell+1}^i \overline{x}_{\ell+1}^i: 1 \le i \le n\} \cup \{c_1 c_2\}\end{math} satisfies the condition.
For an odd $k$, the edge set \begin{math}\{a^i b^i: 1 \le i \le n\} \cup \{c_1 c_2\}\end{math} satisfies the condition.
\end{proof}

\begin{claim}\label{sizen} Assume that there exists an assignment $\phi$ that satisfies $C$. Then there exists a maximal distance-$k$ matching of size $n$ in the graph $G_k$ for any \begin{math}k \ge 2\end{math}.
\end{claim}
\begin{proof}
For an even $k$, for each $i$ such that \begin{math}x_i = 1\end{math} in $\phi$, add the edge  \begin{math}x_1^i x_2^i\end{math} to $M$. For each $i$ such that \begin{math}x_i = 0\end{math} in $\phi$, add the edge \begin{math}\overline{x}_1^i \overline{x}_2^i\end{math} to $M$.

If $k$ is odd, then for each $i$ such that \begin{math}x_i = 1\end{math} in $\phi$, add the edge \begin{math}x_2^i x_2'^i\end{math} to $M$. For each $i$ such that \begin{math}x_i = 0\end{math} in $\phi$, add the edge \begin{math}\overline{x}_2^i \overline{x}_2'^i\end{math} to $M$.

In both cases, if the clause $c_j$ is satisfied, then all the edges in the set $E^\ell_{G_k}(\{c'_j\})$ are at a distance less than $k$ from some edge in $M$. Besides that, for every $i$, any edge in the set $E(G_k[V_i])$ is at a distance less than $k$ from some edge in $M$. It is also easy to see that $M$ is a distance-$k$ matching. Thus, $M$ is a maximal distance-$k$ matching.
\end{proof}

\begin{claim}\label{phi-exists} For any \begin{math}k \ge 2\end{math}, if there exists a maximal distance-$k$ matching $M$ of size $n$ in $G_k$, then there exists an assignment $\phi$ such that no more than one clause in $C$ is not satisfied.
\end{claim}
\begin{proof} By Claim \ref{two-sizes}, all the edges from $M$ are in \begin{math}\cup_i E(G_k[V_i])\end{math}. For each $i$, if there is an edge in $M$ with both ends in $X_i$ (for an even $k$) or in \begin{math}X_i \cup \{a_i\}\end{math} (for an odd $k$), assign \begin{math}x_i = 1\end{math} in $\phi$. For each $i$, if there is an edge in $M$ with both ends in $\overline{X}_i$ (for an even $k$) or in \begin{math}\overline{X}_i \cup \{a_i\}\end{math} (for an odd $k$), assign \begin{math}x_i = 0\end{math} in $\phi$. This assignment is consistent, because $E(G_k[V_i])$ contains at most one edge of $M$ for each $i$. Assign arbitrary values to all the remaining variables.

By the construction of $\phi$, if a clause $c_i$ is not satisfied by the constructed assignment, then there are no edges in $M$ at a distance less than $k$ from $c'_j$. If some distinct clauses $c_p$ and $c_q$ are not satisfied by $\phi$, then the vertices $c'_p$ and $c'_q$ are at a distance at least $k$ from any edge in $M$. Hence, \begin{math}M \cup  \{c'_p c'_q\}\end{math} is a distance-$k$ matching, which contradicts the maximality of $M$. We obtain that $\phi$ does not satisfy at most one clause in $C$.
\end{proof}

\begin{theorem}\label{recogn-main} Recognition of $k$-equimatchable graphs is $\mathrm{co}$-$\mathrm{NP}$-complete for any fixed \begin{math}k \ge 2\end{math}.
\end{theorem}
\begin{proof} It is obvious that this problem is in co-NP. To show that it is NP-hard we use the \textsc{SAT-to-Graph$_k$} reduction from the \textsc{$3$SAT} problem.

Given the set of variables $X$, the multiset of clauses $C$ and an integer number \begin{math}k \ge 2\end{math}, we construct a graph $G_k$ by the algorithm \textsc{SAT-to-Graph$_k$}. Assume that either all the clauses in $C$ can be satisfied by some assignment $\phi$, or for any assignment $\phi$ at least two clauses in $C$ are unsatisfied. Then by Claims \ref{two-sizes}, \ref{sizenp1}, \ref{sizen}, and \ref{phi-exists}, $G$ is not $k$-equimatchable if and only if $C$ is satisfiable.

It remains to show that we can safely assume that there exists no assignment $\phi$ satisfying exactly one clause. Indeed, note that \textsc{$3$SAT} is still NP-complete when each clause in $C$ appears at least twice, since doubling each clause does not affect satisfiability. In this case, each assignment not satisfying at least one clause does not satisfy at least two clauses.
\end{proof}

\begin{corollary} The problem of $k$-equipackable line graph recognition is $\mathrm{co}$-$\mathrm{NP}$-complete for every fixed \begin{math}k \ge 2\end{math}.
\end{corollary}

\begin{theorem}\label{equim-final} The problem of finding a minimum maximal distance-$k$ matching is $\mathrm{NP}$-hard for every fixed \begin{math}k \ge 2\end{math}.
\end{theorem}

\begin{proof} Consider the graph $G_k$ constructed in the proof of Theorem \ref{recogn-main}. By Claim \ref{two-sizes}, $G_k$ is not $k$-equimatchable if and only if $G_k$ has a maximal distance-$k$ matching of size $n$. By Theorem \ref{recogn-main} checking if $G_k$ is not $k$-equimatchable is an NP-complete problem, thus checking whether $G_k$ has a maximal distance-$k$ matching of size less or equal to $n$ is also an NP-complete problem.
\end{proof}

\begin{corollary}
The problem of finding a minimum $k$-independent dominating set in line graphs is $\mathrm{NP}$-hard for every fixed \begin{math}k \ge 2\end{math}.
\end{corollary}

\section{Subclasses of $k$-equipackable graphs}\label{k-equip}

Let $k$ be a positive integer and $\mathcal{R}_k$ be the class of such graphs $G$ that \begin{math}\rho_k(G) = \rho_{2k}(G)\end{math}. Recall that a vertex \begin{math}v \in V(G)\end{math} is called a simplicial vertex if its closed neighborhood $N_G[v]$ induces a complete subgraph of $G$. A {\em simplicial clique} in $G$ is a maximal clique in $G$ containing a simplicial vertex of $G$. Let
\begin{displaymath}\mathcal{S}(G) = \{\text{the set of simplicial vertices of}\ C\ |\ C\ \text{is a simplicial clique of}\ G\} .\end{displaymath}
Then a transversal of $\mathcal{S}(G)$ is by definition a set of simplicial vertices that has exactly one common vertex with every simplicial clique of $G$.

In \cite{DistPack} the following simple characterization of the class $\mathcal{R}_1$ is obtained:

\begin{theorem}
\label{rautenbach}
A graph $G$ satisfies \begin{math}\rho_1(G) = \rho_2(G)\end{math} if and only if the following two statements hold:
\begin{enumerate}
\item[(i)] a subset of $V(G)$ is a maximum $2$-packing if and only if it is a transversal of $\mathcal{S}(G)$, and
\item[(ii)] for every transversal $P$ of $\mathcal{S}(G)$, the sets $N_G[u]$ for $u$ in $P$ partition $V(G)$.
\end{enumerate}
\end{theorem}

Actually, this characterization of class $\mathcal{R}_1$ is slightly superflous in the following sense:
\begin{propos}
    In the statement of Theorem \ref{rautenbach}, condition (ii) implies condition (i).
\end{propos}

This is a direct consequence of the following theorem\footnote{Note that the latter condition in Theorem \ref{cliquepart} is a reformulation of condition (ii) of Theorem \ref{rautenbach}. The presented proof, independent of Theorem \ref{rautenbach}, was kindly contributed by one of the anonymous reviewers.}:

\begin{theorem}
\label{cliquepart}
A graph $G$ satisfies \begin{math}\alpha(G) = \alpha(G^2)\end{math} if and only if the simplicial cliques of $G$ form a partition of $V(G)$.
\end{theorem}
\begin{proof}
    First assume that the simplicial cliques, say \begin{math}Q_1, \ldots, Q_t\end{math}, of $G$ form a partition of $V(G)$. For every $i$, \begin{math}1 \le i \le t\end{math}, let \begin{math}q_i \in Q_i\end{math} be a simplicial vertex in $G$. Clearly, any independent set in $G$ contains at most one vertex of each clique in $G$. Thus \begin{math}\alpha(G) \le t\end{math}. Since for all \begin{math}i \ne j\end{math}, \begin{math}Q_i \cap Q_j = \varnothing\end{math}, the distance between $q_i$ and $q_j$ is at least $3$, and therefore \begin{math}\{q_1, \ldots, q_t\}\end{math} is an independent set in $G^2$. Thus, \begin{math}\alpha(G) = \alpha(G^2)\end{math}.

    Now assume that \begin{math}\alpha(G) = \alpha(G^2)\end{math}. Let \begin{math}I = \{x_1, \ldots, x_s\}\end{math} be a maximum independent set in $G^2$, i.e. \begin{math}\alpha(G^2) = s\end{math}. Clearly, for every $i$, \begin{math}1 \le i \le s\end{math}, $N[x_i]$ is a clique in $G$, since otherwise \begin{math}\alpha(G^2) < \alpha(G)\end{math} holds (replacing the vertex $x_i$ in $I$ with its two nonadjacent neighbours would yield a larger independent set in $G$). Thus, $x_i$ is a simplicial vertex and \begin{math}Q_i = N[x_i]\end{math} is a simplicial clique in $G$. Clearly, for all \begin{math}i \ne j\end{math}, \begin{math}Q_i \cap Q_j = \varnothing\end{math} since $I$ is an independent set in $G^2$. If there was a vertex \begin{math}y \notin Q_1 \cup \ldots \cup Q_s\end{math}, then again \begin{math}\alpha(G^2) < \alpha(G)\end{math}, which is a contradiction. Hence \begin{math}Q_1 \cup \ldots \cup Q_s = V(G)\end{math}. Assume that $Q$ is a simplicial clique of $G$; let \begin{math}q \in Q\end{math} be a simplicial vertex in $G$. Since \begin{math}Q_1 \cup \ldots \cup Q_s = V(G)\end{math}, there is an $i$ with \begin{math}q \in Q_i\end{math}. Now, since $q$ is simplicial, we have \begin{math}Q \subseteq Q_i\end{math}, and since $Q$ and $Q_i$ are inclusion-maximal, we have \begin{math}Q = Q_i\end{math}. Thus, the simplicial cliques of $G$ form a partition of $V(G)$ which proves the theorem.
\end{proof}

As $k$-packings and $2k$-packings in $G$ correspond to $1$-packings and $2$-packings in $G^k$, a graph $G$ belongs to $\mathcal{R}_k$ if and only if \begin{math}G^k \in \mathcal{R}_1\end{math}. This leads to the following structural description of the graphs in $\mathcal{R}_k$.

\begin{corollary}\label{struct-rk} $G$ is a graph with \begin{math}\rho_k(G) = \rho_{2k}(G)\end{math} if and only if there exists a partition \begin{math}V_1, V_2, \ldots, V_t\end{math} of $V(G)$ such that the following two properties hold: \\
1) $V_i$ induces in $G$ a subgraph of diameter $k$ for \begin{math}1 \le i \le t\end{math};\\
2) $V_i$ contains a vertex $v_i$ such that for each edge \begin{math}vu \in E(G), v \in V_i, u \not \in V_i\end{math}, the distance between $v_i$ and $v$ equals $k$ for \begin{math}1 \le i \le t\end{math}.
\end{corollary}

\begin{corollary}\label{rautenbach-equi} Every graph in $\mathcal{R}_k$ is $k$-equipackable.
\end{corollary}

\begin{proof}
We start by proving the claim for $k=1$. We know that for every \begin{math}G \in \mathcal{R}_1\end{math}, the set of simplicial cliques partitions $V(G)$. Let $I$ be any maximum independent set of $G$, and let $t$ be the number of its simplicial cliques. Then both \begin{math}|I| \leq t\end{math} (because every clique can contain no more than one vertex of an independent set) and \begin{math}|I| \geq t\end{math} (because if there exists a simplicial clique containing no vertices of an independent set, the simplicial vertex from the clique can extend the independent set).

Now, if \begin{math}G \in \mathcal{R}_k\end{math} for \begin{math}k > 1\end{math}, the claim of the theorem is implied by \begin{math}G^k \in \mathcal{R}_1\end{math} and thus by the 1-equipackability (well-coveredness) of $G^k$.
\end{proof}

Applying the structural characterization of Corollary \ref{struct-rk} to line graphs, one can obtain a characterization of graphs with equal maximum distance-$k$ and distance-$2k$ matching sizes given by \cite{DistPack}. Corollary \ref{rautenbach-equi} implies that such graphs constitute a subclass of $k$-equimatchable graphs.

\begin{corollary}
Every graph with equal maximum distance-$k$ and distance-$2k$ matching sizes is $k$-equimatchable.
\end{corollary}

A {\em sun} is a chordal graph \begin{math}G = (V, E)\end{math} such that $V$ can be partitioned into two sets \begin{math}X = \{x_i : 1 \le i \le n\}\end{math} and \begin{math}Y = \{y_i : 1 \le i \le n\}\end{math} so that

1) $X$ is an independent set in $G$,

2) \begin{math}y_iy_{i+1} \in E(G), 1 \le i \le n - 1\end{math}, \begin{math}y_ny_1 \in E(G)\end{math},

3) \begin{math}x_iy_j \in E(G)\end{math} if and only if \begin{math}j = i\end{math}, or \begin{math}j = i + 1\end{math}, or \begin{math}i = n, j = 1\end{math}.

A chordal graph is called {\em sun-free chordal} [\cite{sun-free}] if it does not contain any sun as an induced subgraph. This is actually one of the characterizations of strongly chordal graphs [\cite{strongly-chordal}].

In \cite{sun-free} it has been proved that for every sun-free chordal graph $G$ the equality \begin{math}\gamma_k(G) = \rho_{2k}(G)\end{math} holds. Thus, we obtain the following simple characterization of strongly chordal graphs with equal $k$-domination and $k$-packing numbers. According to Corollary \ref{struct-rk}, their structure is analogous to the structure of the trees with equal $k$-domination and $k$-packing numbers obtained by \cite{ToppVolkman} and, more generally, to the structure of block graphs with the same property obtained by \linebreak\cite{wdblock}.

\begin{theorem}\label{sfc} The following statements are equivalent for a strongly chordal graph $G$: \\
1) \begin{math}\gamma_k(G) = \rho_k(G)\end{math};\\
2) \begin{math}G \in \mathcal{R}_k\end{math}.
\end{theorem}

Consider that for an arbitrary graph $G$, both statements in Theorem \ref{sfc} independently imply that $G$ is $k$-equipackable: the equality \begin{math}\gamma_k(G) = \rho_k(G)\end{math} implies \begin{math}i_k(G) = \rho_k(G)\end{math} because \begin{math}\gamma_k(G) \le i_k(G) \le \rho_k(G)\end{math}; and
\begin{math}G \in \mathcal{R}_k\end{math} implies that $G$ is $k$-equipackable by Corollary \ref{rautenbach-equi}. Theorem \ref{sfc} establishes that in the class of strongly chordal graphs, the sets of $k$-equipackable graphs defined by these two statements coincide. The following questions arise in connection with this result:

1) Are there any $k$-equipackable strongly chordal graphs besides those satisfying the statements of Theorem \ref{sfc}?

2) Can the equivalence of the statements mentioned in Theorem \ref{sfc} be extended beyond the class of chordal graphs?

\section[Weighted minimum maximal distance-$k$ matchings and $k$-independent dominating sets in chordal graphs]{Weighted minimum maximal distance-$k$ matchings\\and $k$-independent dominating sets in chordal graphs}\label{non-approx}

The \textsc{Weighted Minimum Maximal Distance-$k$ Matching} problem (which we shall refer to as \textsc{$k$-WMMM}) is the problem of finding a minimum weight maximal distance-$k$ matching in a given graph $G$ with a weight function \begin{math}w: E(G) \rightarrow \mathbb{Q}_+\end{math} defined on its edges, where $\mathbb{Q}_+$ is the set of positive rational numbers. The \textsc{Set Cover} problem is the problem of finding for the given universal set \begin{math}X = \left\{X_1, \ldots, X_{|X|}\right\}\end{math} and a collection \begin{math}S = \left\{S_1, \ldots, S_{|S|}\right\}\end{math} of its subsets a minimum number of sets from $S$ that cover $X$ [\cite{Garey}]. A family \begin{math}C \subseteq S\end{math} of sets is said to {\em cover} $X$ if each element of $X$ is contained in at least one set of $C$.

We shall build a construction connecting approximate solutions to the aforementioned problems for \begin{math}k = 2 \ell\end{math}. Assume we are given integer numbers \begin{math}\ell > 1\end{math} and \begin{math}d > 1\end{math}, a positive rational number $c$ and an instance \begin{math}(X, S)\end{math} of the \textsc{Set Cover} problem such that \begin{math}|X| \ge 5\ell+2\end{math}, \begin{math}c(\ln |X| - 1) \ge 1\end{math} and \begin{math}|S| \le |X|^d\end{math}. We construct the graph \begin{math}G = G(\ell, d, c, X, S) = (V, E)\end{math} as follows. Take \begin{math}V = X \cup \{v_{i,j} \mid 1 \le i \le |S|, 1 \le j \le 5\ell+1\}\end{math}, as the vertex set (where all the $v_{i,j}$ are new vertices). The edge set of $G$ is defined by

\begin{displaymath}E = \{x_i x_j \mid x_i, x_j \in X, i \ne j\} \cup \{x_i v_{j,1} \mid x_i \in S_j\} \cup \{v_{i,j} v_{i,j+1}\}\end{displaymath}
(see Figure \ref{kr.fig.2l-wmmm}).

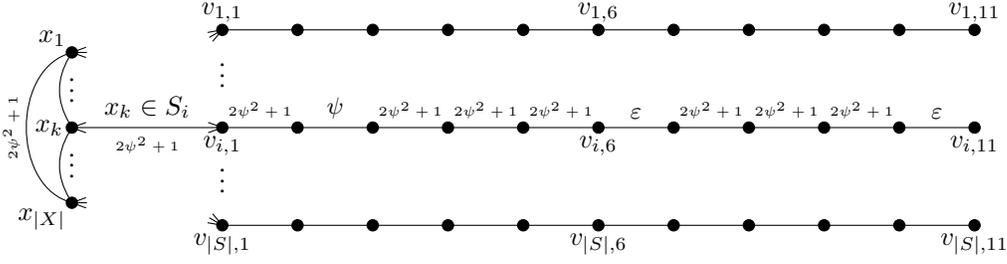
\begin{figure}[!ht]
	\centering%
	\usetikzlibrary{fit,graphs,quotes,shapes}
\begin{tikzpicture}[
    vertex/.style={circle, draw=black, fill, minimum width=1.5mm, inner sep=0pt, outer sep=0pt},
    every label/.style={inner sep=0pt, minimum width=0pt, label distance=0.3mm},
]
\graph[nodes=vertex, empty nodes, no placement] {
    {
        [edges=bend right,x=-1]
        x1[y=1,label=above left:$x_1$] -- xk[y=0,label=left:$x_k$] -- xx[y=-1,label=below left:$x_{|X|}$]
            --[bend left=70] x1;
    };
    {
        [y=1.3,/tikz/label position=above]
        v11[x=1,label={[name=v11label]$v_{1,1}$}] --
        v12[x=2] --
        v13[x=3] --
        v14[x=4] --
        v15[x=5] --
        v16[x=6,label={[name=v16label]$v_{1,6}$}] --
        v17[x=7] --
        v18[x=8] --
        v19[x=9] --
        v110[x=10] --
        v111[x=11,label={[name=v111label]$v_{1,11}$}];
    };
    {
        [y=0,/tikz/label position=below]
        vi1[x=1,label={[name=vi1label]$v_{i,1}$}] --["\tiny$2\psi^2 + 1$"]
        vi2[x=2] --["\small$\psi$"]
        vi3[x=3] --["\tiny$2\psi^2 + 1$"]
        vi4[x=4] --["\tiny$2\psi^2 + 1$"]
        vi5[x=5] --["\tiny$2\psi^2 + 1$"]
        vi6[x=6,label={[name=vi6label]$v_{i,6}$}] --["\small$0$"]
        vi7[x=7] --["\tiny$2\psi^2 + 1$"]
        vi8[x=8] --["\tiny$2\psi^2 + 1$"]
        vi9[x=9] --["\tiny$2\psi^2 + 1$"]
        vi10[x=10] --["\small$0$"]
        vi11[x=11,label={[name=vi11label]$v_{i,11}$}];
    };
    {
        [y=-1.3,/tikz/label position=below]
        vs1[x=1,label={[name=vs1label]$v_{|S|,1}$}] --
        vs2[x=2] --
        vs3[x=3] --
        vs4[x=4] --
        vs5[x=5] --
        vs6[x=6,label={[name=vs6label]$v_{|S|,6}$}] --
        vs7[x=7] --
        vs8[x=8] --
        vs9[x=9] --
        vs10[x=10] --
        vs11[x=11,label={[name=vs11label]$v_{|S|,11}$}];
    };
};
\node[rotate=90] at (-1.8, 0) {\tiny $2 \psi^2 + 1$};

\draw (xk) --node[above] {$x_k \in S_i$} node[below] {\tiny $2 \psi^2 + 1$} (vi1);
\draw (v11) -- +(-170:0.2) (v11) -- +(-150:0.2) (v11) -- +(-130:0.2);
\draw (vi1) -- +(-160:0.2) (vi1) -- +(160:0.2);
\draw (vs1) -- +(170:0.2) (vs1) -- +(150:0.2) (vs1) -- +(130:0.2);

\draw (x1) -- +(20:0.2) (x1) -- +(0:0.2) (x1) -- +(-20:0.2);
\draw (xk) -- +(20:0.2) (xk) -- +(-20:0.2);
\draw (xx) -- +(20:0.2) (xx) -- +(0:0.2) (xx) -- +(-20:0.2);

\node at (-1, -0.4) {$\vdots$};
\node at (-1, 0.6) {$\vdots$};

\node at (1, -0.6) {$\vdots$};
\node at (1, 0.8) {$\vdots$};
\end{tikzpicture}%
	\caption{The construction used in the reduction of \textsc{Set Cover} to \textsc{4-WMMM}\label{kr.fig.2l-wmmm}}
\end{figure}

It is easy to see that $G$ has no induced cycles of length $4$ or more, thus $G$ is chordal.

As \begin{math}|S| \le |X|^d\end{math} and \begin{math}|X| \ge 5\ell + 2\end{math}, we have
\begin{displaymath}|V| = |X| + (5\ell+1)|S| \le |X| + (5\ell+1)|X|^d \le |X|^{d+1}.\end{displaymath}

Let $p$ be the maximum integer number less than \begin{math}\ln |V|\end{math}. Assume without loss of generality that \begin{math}p > 1\end{math}. We have \begin{math}c p \ge c (\ln |X| - 1) \ge 1\end{math}. Take \begin{math}\psi = c p |S|  \ge |S|\end{math}. Assign a weight of $\psi$ to the edges $v_{i,\ell}v_{i,\ell+1}$, \begin{math}1 \le i \le |S|\end{math}. Assign a weight of zero to the edges $v_{i,3\ell}v_{i,3\ell+1}$ and $v_{i,5\ell}v_{i,5\ell+1}$, \begin{math}1 \le i \le |S|\end{math}. Assign a weight of \begin{math}2 \psi^2 + 1\end{math} to all the remaining edges. Let $\omega$ denote the weight function defined this way.

Clearly, it is possible to build \begin{math}G(\ell, d, c, X, S)\end{math} in time polynomial in $|X|$. The next theorem shows how to use $G$ as an approximation preserving reduction of the \textsc{Set Cover} problem to the \textsc{$2\ell$-WMMM} problem in chordal graphs.

\begin{theorem}
\label{WMMM}
Assume that $M^*$ is a minimum weight maximal distance-$2\ell$ matching in the chordal graph \begin{math}G(\ell, d, c, X, S)\end{math}. Let $D^*$ be a set of such indices $i$ that \begin{math}\left\{S_i \mid i \in D^*\right\}\end{math} is a minimum set cover of $X$ with respect to $S$. Let $M_A$ be any distance-$2 \ell$ matching in $G$ such that \begin{math}\omega(M_A) / \omega(M^*) \le c \ln |V(G)|\end{math}. Then a set of indices $D_{M_A}$ such that \begin{math}\left\{S_i \mid i \in D_{M_A}\right\}\end{math} cover $X$ and \begin{math}|D_{M_A}| / |D^*| \le c (d + 1) \ln |X|\end{math} can be computed given $M_A$ in time polynomial in $|X|$.
\end{theorem}
\begin{proof}
No edge of weight \begin{math}2 \psi^2 + 1\end{math} can occur in $M^*$, because the set
\begin{displaymath}\{v_{i,\ell} v_{i,\ell+1}, v_{i,5\ell} v_{i,5\ell+1} : 1 \le i \le |S|\}\end{displaymath}
is a maximal distance-$2\ell$ matching of weight \begin{math}w_0 = (\psi + 0) |S| \ge \omega(M^*)\end{math} and \begin{math}w_0 < 2 \psi^2 + 1\end{math}. Neither can such an edge occur in $M_A$, because otherwise we would have \begin{displaymath}2 \psi^2 + 1 \le \omega(M_A) \le \left(c \ln |V|\right) \omega(M^*) < c (p + 1) \cdot (\psi + 0) |S| < 2 c p |S| \psi = 2 \psi^2,\end{displaymath}
where the second inequality follows from the approximation quality of $M_A$, the third employs $p \geq \ln |V| - 1$ following from the definition of $p$, and the fourth follows from the assumption $p > 1$ giving $p + 1 < 2p$.

Consider an arbitrary maximal distance-$2\ell$ matching $M$ with no edges of weight \begin{math}2 \psi^2 + 1\end{math} in $G$. Each simple path \begin{math}L_i = v_{i,1} v_{i,2} \ldots v_{i,5\ell+1}\end{math} can contain exactly one or two edges of $M$. Moreover, either \begin{math}L_i \cap M = \{v_{i,3\ell} v_{i,3\ell+1}\}\end{math} or \begin{math}L_i \cap M = \{v_{i,\ell}v_{i,\ell+1}, v_{i,5\ell}v_{i,5\ell+1}\}\end{math}.

Let $D_M$ denote the set of such indices $i$ that each path $L_i$ contains exactly two edges of $M$. Then \begin{math}\omega(M) = \psi |D_M|\end{math}. We shall prove that the collection of sets \begin{math}S_M = \{S_i \mid i \in D_M\}\end{math} is a cover of the set $X$.

Assume the contrary. Let $x_t$ be an element that is not contained in any set from $S_M$. Then for every $i$ such that \begin{math}x_t \in S_i\end{math}, $L_i$ contains only one edge, namely \begin{math}v_{i,3\ell} v_{i,3\ell+1}\end{math}. Therefore \begin{math}M \cup \{x_t v_{i,1}\}\end{math} (for \begin{math}\ell = 1\end{math}) or \begin{math}M \cup \{v_{i,\ell-1}v_{i,\ell}\}\end{math} (for \begin{math}\ell \ge 2\end{math}) is a distance-$2\ell$ matching, contradicting the maximality of $M$.

On the other hand, as $D^*$ is a set of such indices $i$ that \begin{math}\left\{S_i \mid i \in D^*\right\}\end{math} is a minimum set cover of $X$, the set
\begin{displaymath}\{v_{i,\ell} v_{i,\ell+1}, v_{i,5\ell} v_{i,5\ell+1} \mid i \in D^*\} \cup \{v_{i,3\ell} v_{i,3\ell+1} \mid i \notin D^* \}\end{displaymath}
is a distance-$2\ell$ matching having the weight \begin{math}\psi |D^*| + 0 \cdot |S|\end{math}. Thus, \begin{math}\omega(M^*) = \psi |D^*|\end{math}.

Then
\begin{displaymath}\psi |D_{M_A}| = \omega(M_A) \le \left(c \ln |V|\right) \omega(M^*) = \left(c \ln |V|\right) \psi |D^*|,\end{displaymath}
so, since \begin{math}\ln |V| \le (d + 1) \ln |X|\end{math}, we have
\begin{displaymath}\psi |D_{M_A}| \le \left( c \ln |V| \right) \psi |D^*| \le \psi \cdot c (d + 1) \ln |X| \cdot |D^*|.\end{displaymath}
This provides \begin{math}|D_{M_A}| \le c(d+1) \ln |X| \cdot |D^*|\end{math}. It is easy to see that given $M_A$, $|D_{M_A}|$ is straightforward to obtain in time polynomial in $|X|$.
\end{proof}

In \cite{Dinur} it is shown that for every \begin{math}0 < z < 1\end{math} it is NP-hard to approximate the \textsc{Set Cover} problem within a factor of \begin{math}(1 - z) \ln n\end{math}, where $n$ is the cardinality of the universal set $X$. The instances of the \textsc{Set Cover} problem used in the proof of this result have no more than $n^{O(1/z)}$ subsets in $S$ [\cite{Dinur}]. Combined with Theorem \ref{WMMM} it implies the following.

\begin{corollary}
\label{WMMMcor}
For some positive rational constant $c$, it is \textrm{NP}-hard to approximate the problem \textsc{$2\ell$-WMMM} within a factor of \begin{math}c \ln |V(G)|\end{math} in chordal graphs.
\end{corollary}
\begin{proof}
Fixing an arbitrary $z$, \begin{math}0 < z < 1\end{math}, we may assume that the problem \textsc{Set Cover} is NP-hard to approximate within a factor of \begin{math}(1 - z) \ln |X|\end{math} with the restriction \begin{math}|S| \leq |X|^d\end{math} for some integer number $d$, where $X$ is the universal set and $S$ is the collection of its subsets [\cite{Dinur}].

But taking \begin{math}c = (1 - z)	/ (d + 1)\end{math}, building \begin{math}G(\ell, d, c, X, S)\end{math} and applying Theorem \ref{WMMM}, we can transform in polynomial time a \begin{math}c \ln |V(G)|\end{math}-approximation of \textsc{$2\ell$-WMMM} in this graph into a \begin{math}(1 - z) \ln |X|\end{math}-approximation of \textsc{Set Cover} subject to the mentioned restriction.
\end{proof}

Note that in the construction of \begin{math}G(\ell, d, c, X, S)\end{math} we are using only $3$ different edge weights upper bounded by a polynomial of $|S|$. Therefore the result of Theorem \ref{WMMM} still holds for the \textsc{$2\ell$-WMMM} problem with integer polynomially bounded edge weights (it is sufficient to take \begin{math}\lceil \psi \rceil\end{math} instead of $\psi$ and multiply all the weights in the proof by~$2 \psi$ to show this).

The idea of the proof of Theorem \ref{WMMM} can be applied to the vertex counterpart of the \textsc{$k$-WMMM} problem,
\textsc{Weighted $k$-Independent Dominating Set} (\textsc{$k$-WIDS} for short), which is the problem of finding a minimum weight $k$-independent dominating set in a given graph $G$ with a weight function \begin{math}w: V(G) \rightarrow \mathbb{Q}_+\end{math} defined on its vertices.

\begin{theorem}
    For some positive rational constant $c$, it is \textrm{NP}-hard to approximate the problem \textsc{$(2\ell - 1)$-WIDS} within a factor of \begin{math}c \ln |V(G)|\end{math} in chordal graphs.
\end{theorem}
\begin{proof}
By analogy with the proof of Theorem \ref{WMMM} and Corollary \ref{WMMMcor}, we will assume that for an arbitrary fixed positive integer number $\ell$ and some rational constant $c$ there exists a polynomial time \begin{math}c \ln |V(G)|\end{math}-approximation algorithm \textsc{\begin{math}(2\ell - 1)\end{math}-WIDS-Approx} for the problem \textsc{\begin{math}(2\ell - 1)\end{math}-WIDS} in chordal graphs, and show that for any integer number \begin{math}d > 1\end{math} there exists a polynomial time \begin{math}c(d+1) \ln |X|\end{math}-approximation algorithm for the \textsc{Set Cover} problem with the restriction \begin{math}|S| \le |X|^d\end{math}.

Consider the construction from the proof of Theorem \ref{WMMM} (Figure \ref{kr.fig.2l-wmmm}) but with chains of length \begin{math}5 \ell - 3\end{math} (formed by the vertices \begin{math}\{v_{i,j} \mid 1 \le i \le |S|, 1 \le j \le 5 \ell - 2\}\end{math}, Figure \ref{kr.fig.2lm1-wids}).

\begin{figure}[!ht]
\centering%
\usetikzlibrary{fit,graphs,quotes,shapes}
\begin{tikzpicture}[
    vertex/.style={circle, draw=black, fill, minimum width=1.5mm, inner sep=0pt, outer sep=0pt},
    every label/.style={inner sep=0pt, minimum width=0pt, label distance=0.3mm},
]
\graph[nodes=vertex, empty nodes, no placement] {
    {
        [edges=bend right,x=-1.5]
        x1[y=1,label=above:$x_1$,label={[label distance=0.3cm]\tiny$2\psi^2 + 1$}] -- xk[y=0,label=left:$x_k$] -- xx[y=-1,label=below:$x_{|X|}$,label={[label distance=0.4cm]below:\tiny$2\psi^2 + 1$}]
            --[bend left=70] x1;
    };
    {
        [y=1.3,/tikz/label position=above]
        v11[x=1,label={$v_{1,1}$}] --
        v12[x=2] --
        v13[x=3] --
        v14[x=4] --
        v15[x=5,label={$v_{1,5}$}] --
        v16[x=6] --
        v17[x=7] --
        v18[x=8,label={$v_{1,8}$}];
    };
    {
        [y=0,/tikz/label position=above]
        vi1[x=1,label=below:$v_{i,1}$,label={[label distance=0.8mm]\tiny$2\psi^2 + 1$}] --
        vi2[x=2,label={[label distance=0.8mm]\small$\psi$}] --
        vi3[x=3,label={[label distance=0.8mm]\tiny$2\psi^2 + 1$}] --
        vi4[x=4,label={[label distance=0.8mm]\tiny$2\psi^2 + 1$}] --
        vi5[x=5,label=below:$v_{i,5}$,label={[label distance=0.8mm]\small$0$}] --
        vi6[x=6,label={[label distance=0.8mm]\tiny$2\psi^2 + 1$}] --
        vi7[x=7,label={[label distance=0.8mm]\tiny$2\psi^2 + 1$}] --
        vi8[x=8,label=below:$v_{i,8}$,label={[label distance=0.8mm]\small$0$}];
    };
    {
        [y=-1.3,/tikz/label position=below]
        vs1[x=1,label={$v_{|S|,1}$}] --
        vs2[x=2] --
        vs3[x=3] --
        vs4[x=4] --
        vs5[x=5,label={$v_{|S|,5}$}] --
        vs6[x=6] --
        vs7[x=7] --
        vs8[x=8,label={$v_{|S|,8}$}];
    };
};
\draw (xk) --node[above] {$x_k \in S_i$} (vi1);
\draw (v11) -- +(-170:0.2) (v11) -- +(-150:0.2) (v11) -- +(-130:0.2);
\draw (vi1) -- +(-160:0.2) (vi1) -- +(160:0.2);
\draw (vs1) -- +(170:0.2) (vs1) -- +(150:0.2) (vs1) -- +(130:0.2);

\draw (x1) -- +(20:0.2) (x1) -- +(0:0.2) (x1) -- +(-20:0.2);
\draw (xk) -- +(20:0.2) (xk) -- +(-20:0.2);
\draw (xx) -- +(20:0.2) (xx) -- +(0:0.2) (xx) -- +(-20:0.2);

\node at (-1.5, -0.4) {$\vdots$};
\node at (-1.5, 0.6) {$\vdots$};

\node at (1, -0.6) {$\vdots$};
\node at (1, 0.8) {$\vdots$};
\end{tikzpicture}%
\caption{The construction used in the reduction from \textsc{Set Cover} to \textsc{3-WIDS}\label{kr.fig.2lm1-wids}}
\end{figure}
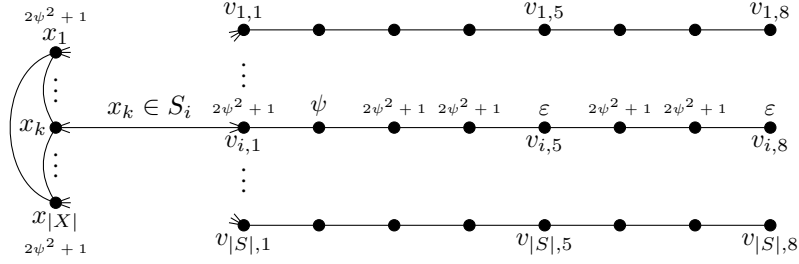

The weights are now assigned to the vertices (instead of the edges) as follows: the weight of $\psi$ is assigned to $v_{i,\ell}$, \begin{math}1 \le i \le |S|\end{math}, zero weight is assigned to $v_{i,3\ell-1} $ and $v_{i,5\ell-2}$, \begin{math}1 \le i \le |S|\end{math}, and all the remaining vertices are assigned the weight of \begin{math}2\psi^2 + 1\end{math}.

The rest of the proof is analogous to that of Theorem \ref{WMMM} and Corollary \ref{WMMMcor} and is therefore omitted.
\end{proof}

\section{Distance problems in large-girth planar graphs}\label{sec-restr}

Now we proceed to some hardness results for \textsc{Minimum Maximal Distance-$2$ Matching} (which is also known under the name of \textsc{Minimum Maximal Induced Matching}), the problem of finding a maximal distance-$2$ matching of minimum cardinality in a given graph. We consider the following supplementary transformation $T(G)$ applied to an arbitrary graph $G$: For every vertex $v$ in $G$, introduce the new vertices $a^v$, $b^v$, $c^v$, $d^v_i$, $e^v_i$ and the new edges \begin{math}v a^v\end{math}, \begin{math}a^v b^v\end{math}, \begin{math}v c^v\end{math}, \begin{math}c^v d^v_i\end{math}, \begin{math}d^v_i e^v_i\end{math} for \begin{math}1 \le i \le 2n+1\end{math}, where \begin{math}n = |V(G)|\end{math} (see Figure \ref{kr.fig.gadget}).

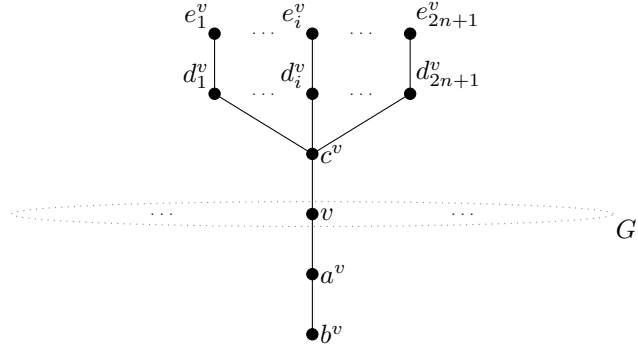
\begin{figure}[!ht]
\centering%
\begin{tikzpicture}[
    vertex/.style={circle, draw=black, fill, minimum width=1.5mm, inner sep=0pt, outer sep=0pt},
    every label/.style={inner sep=0pt, minimum width=0pt, label distance=0.2mm},
    yscale=0.8
]
\graph[nodes=vertex, empty nodes, no placement] {
    {
        [x=0]
        bv[y=-2,label=right:$b^v$] --
        av[y=-1,label=right:$a^v$] --
        v[y=0,label=right:$v$] --
        cv[y=1,label=right:$c^v$] -- 
        dvi[y=2,label=above left:$d^v_i$] --
        evi[y=3,label=above left:$e^v_i$];
    };
    {
        [x=-1.3]
        cv -- dv1[y=2,label=above left:$d^v_1$] -- ev1[y=3,label=above left:$e^v_1$];
    };
    {
        [x=1.3]
        cv -- dvl[y=2,label=above right:$d^v_{2n+1}$] -- evl[y=3,label=above right:$e^v_{2n+1}$];
    };
};
\node at (v) [draw,dotted,gray,ellipse,minimum width=8cm,label={[yshift=-2mm]right:$G$}] {};
\node at (-0.65, 2) {\tiny$\ldots$};
\node at (-0.65, 3) {\tiny$\ldots$};
\node at (0.65, 2) {\tiny$\ldots$};
\node at (0.65, 3) {\tiny$\ldots$};
\node at (-2, 0) {\tiny$\ldots$};
\node at (2, 0) {\tiny$\ldots$};
\end{tikzpicture}%
\caption{The construction of the graph $T(G)$\label{kr.fig.gadget}}
\end{figure}

Let $\sigma(G)$ denote the cardinality of a minimum maximal induced matching in $G$. Recall that \begin{math}\alpha(G) = \rho_1(G)\end{math} is the cardinality of a maximum independent set in $G$.

\begin{lemma} Let $G$ be an arbitrary graph. Then \begin{math}\sigma(T(G)) = 2n - \alpha(G)\end{math}.
\label{sigmaAlpha}
\end{lemma}
\begin{proof}
Let us prove that \begin{math}\sigma(T(G)) \le 2n - \alpha(G)\end{math}. Given a maximum independent set $I$ in $G$, build a maximal induced matching $M$ in $T(G)$ as follows: \begin{displaymath}M = \{v c^v \ |\ v \in I\} \cup \{a^v b^v, c^v d^v_1 \ |\ v \in V(G) \setminus I\}.\end{displaymath}
As \begin{math}|M| = \sum\limits_{v\in I} 1 + \sum\limits_{v\in V(G) \setminus I} 2 = 2n - |I| = 2n - \alpha(G)\end{math}, it proves the required inequality.

Let us now prove that \begin{math}\alpha(G) \ge 2n - \sigma(T(G))\end{math}. Consider a minimum maximal induced matching $M$ in $T(G)$. Note that if for some vertex \begin{math}v \in V(G)\end{math} the set $M$ contains any of the edges \begin{math}d^v_i e^v_i\end{math} (for some $i$), \begin{math}v a^v\end{math} or $uv$ for some \begin{math}u \in V(G)\end{math}, then by maximality it must also contain all the edges \begin{math}d^v_j e^v_j\end{math}, \begin{math}1 \le j \le 2n+1\end{math}. But then \begin{math}\sigma(T(G)) = |M| \ge 2n + 1\end{math}, which contradicts the inequality proven before.

Therefore, for every vertex \begin{math}v \in V(G)\end{math}, $M$ contains either a single edge \begin{math}v c^v\end{math} or two edges \begin{math}a^v b^v\end{math} and \begin{math}c^v d^v_i\end{math} for some $i$. The vertices of the first type induce an independent set in $G$, because $M$ is an induced matching. Thus, \begin{math}\alpha(G) + |M| \ge 2n\end{math}, and further \begin{math}\alpha(G) \ge 2n - |M| = 2n - \sigma(T(G))\end{math}, as required.
\end{proof}

To proceed, we need the following well-known result.

\begin{theorem}[\cite{Murphy}]
\label{Murphy}
For arbitrary fixed constants \begin{math}c > 0\end{math} and \begin{math}0 \le r < 1\end{math}, the \textsc{Maximum Independent Set} problem is $\mathrm{NP}$-hard when restricted to planar graphs $G$ with vertices of degrees $2$ and $3$ and containing no cycles of length less than $cn^r$, where \begin{math}n = |V(G)|\end{math}.
\end{theorem}

\begin{theorem}
For arbitrary fixed constants \begin{math}a > 0\end{math} and \begin{math}0 \le d < \frac{1}{3}\end{math}, the \textsc{Minimum Maximal Distance-$2$ Matching} problem is $\mathrm{NP}$-hard in the class of planar graphs $H$ with girth at least $aN^d$, where \begin{math}N = |V(H)|\end{math}.
\end{theorem}
\begin{proof}
We shall build a polynomial time reduction from the \textsc{Maximum Independent Set} problem subject to the restrictions of Theorem \ref{Murphy}.
Let \begin{math}a > 0\end{math} and \begin{math}0 \le d < \frac{1}{3}\end{math} be fixed for the course of the proof.
Choose \begin{math}c = 2^d a\end{math} and \begin{math}r = 3d\end{math} and observe that they satisfy the constraints of Theorem \ref{Murphy}.

Given a planar graph $G$ with girth at least $cn^r$, where $n=|V(G)|$, we can clearly build \begin{math}H = T(G)\end{math} in polynomial time. Observe that $H$ is also planar, has no new cycles and \begin{math}N = |V(H)| = (3 + 2(n+1))n = 2n^2 + 5n\end{math}. Then, for a sufficiently large $n$, \begin{math}2n^3 > N\end{math}. Every cycle of $G$ (and therefore of $H$) has length at least \begin{math}cn^r = a (2n^3)^d > aN^d\end{math}. By Lemma \ref{sigmaAlpha}, an independent set of size at least $s$ exists in $G$ if and only if $H$ contains a maximal induced matching of size no more than \begin{math}2n - s\end{math}.
\end{proof}

We also prove, using a very simple reduction, that the \textsc{Minimum Independent Dominating Set} problem is $\mathrm{NP}$-complete in planar graphs of large girth.

\begin{theorem}
For arbitrary fixed constants \begin{math}a > 0\end{math} and \begin{math}0 \le d < 1\end{math}, the \textsc{Minimum Independent Dominating Set} problem is $\mathrm{NP}$-hard in the class of planar graphs $G$ of degree at most $5$ and girth at least $aN^d$, where $N$ is the number of vertices of~$G$.
\end{theorem}
\begin{proof}
We again build a polynomial time reduction from the the \textsc{Maximum Independent Set} problem subject to the restrictions of Theorem \ref{Murphy}. Assume \begin{math}a > 0\end{math} and \begin{math}0 \le d < 1\end{math} are fixed thoughout the proof. Choose \begin{math}c = 3^d a\end{math} and \begin{math}r = d\end{math}, thus satisfying the constraints of Theorem \ref{Murphy}.

Let \begin{math}G = (V, E)\end{math} be a planar graph with girth at least $cn^r$ and degree at most $3$ provided as the input of \textsc{Maximum Independent Set}, where $n=|V|$. Construct the graph \begin{math}G' = (V', E')\end{math} as follows. Take \begin{math}V' = V \cup \{a_i, b_i \mid 1 \le i \le n\}\end{math}, where \begin{math}n = |V|\end{math}. Take

\begin{displaymath}E' = E \cup \{v_i a_i, v_i b_i \mid 1 \le i \le n\}.\end{displaymath}

Less formally, the transformation consists of attaching two distinct endvertices to each vertex.

Observe that \begin{math}N = |V'| = 3n\end{math}. Every cycle of $G'$ is also a cycle of $G$ and therefore has length at least \begin{math}cn^r = \dfrac{c}{3^d}(3n)^d = aN^d\end{math}.

We claim that \begin{math}\alpha(G) + i(G') = 2n\end{math}. Indeed, let $I$ be a maximum independent set in $G$. Then \begin{math}I' = I \cup \{a_i, b_i \mid v_i \in V \setminus I\}\end{math} is an independent dominating set in $G'$. Thus, \begin{math}|I'| = \alpha(G) + 2 (n - \alpha(G)) = 2n - \alpha(G) \ge i(G')\end{math}, or \begin{math}\alpha(G) + i(G') \le 2n\end{math}.

Next, let $I'$ be a minimum independent dominating set in $G'$. Then each set \begin{math}\{v_i, a_i, b_i\}\end{math} contains either exactly one vertex of $I'$ (and then it is $v_i$) or exactly two vertices of $I'$ (and therefore they are $a_i$ and $b_i$). Then \begin{math}|I'| = |V \cap I'| + 2|V \setminus I'| = 2|V| - |V \cap I'| = 2n - |I|\end{math}, where \begin{math}I = V \cap I'\end{math} is an independent set in $G$ of size
\begin{math}|I| = 2n - |I'| = 2n - i(G')\end{math}. Hence, \begin{math}\alpha(G) \ge 2n - i(G')\end{math}, or \begin{math}\alpha(G) + i(G') \le 2n\end{math}, which completes the proof.
\end{proof}

\section{Concluding remarks}\label{sec-concl}
In Section \ref{k-equip}, we have shown that the class $\mathcal{R}_k$ from \cite{DistPack} is an interesting subclass of the class of $k$-equipackable graphs, a natural extension of well-covered graphs. Applying this result to line graphs allows to show that every graph with equal maximum distance-$k$ and distance-$2k$ matching sizes is $k$-equimatchable. It is notable that belonging to $\mathcal{R}_k$ for a strongly chordal graph is equivalent to having equal $k$-domination and $k$-packing numbers. It is an open question whether there are any other $k$-equipackable strongly chordal graphs or whether this duality result can be extended to some larger class of graphs.

We have shown that the \textsc{Distance-$k$ Weighted Minimum Maximal Matching} problem is NP-hard to approximate in the class of chordal graphs within a factor of \begin{math}c \ln n\end{math}, where $n$ is the number of the input graph vertices, for some fixed constant $c$ and every fixed even $k$ (Section \ref{non-approx}). It thus remains an open problem to obtain corresponding approximation bounds in chordal graphs for odd values of $k$. The same question applies to the \textsc{Weighted $k$-Independent Dominating Set} problem for even $k$. It is also open whether the obtained inapproximability bounds are tight.

The hardness result for the \textsc{Minimum Maximal Induced Matching} problem in large-girth planar graphs from Section \ref{sec-restr} is obtained using a construction with the maximum vertex degree more than the number of the input graph vertices. It is thus interesting whether these results hold when the vertex degrees are limited. Further research might also concern the case of minimum maximal distance-$k$ matchings in large-girth graphs for higher values of $k$.

\section*{Acknowledgements}
We thank Yury Orlovich for stating the initial problems that led to the main results of Sections \ref{k-equim} and \ref{non-approx}, as long as for useful comments and suggestions, and Dieter Rautenbach for a discussion about graphs with equal $k$-packing and $2k$-packing numbers (Section \ref{k-equip}). The technique used in Section \ref{sec-restr} is inspired by an unpublished work of Oleg Duginov and Eugene Dolzhenok on the dissociation sets problem. We also thank anonymous reviewers for their numerous helpful comments on the preliminary versions of this paper.

This work has been partially supported by the Belarusian BRFFR grant (Project F15MLD-022).

\nocite{*}
\bibliographystyle{abbrvnat}

\end{document}